\documentclass[11pt]{article}

\usepackage[margin=1in]{geometry}
\usepackage[T1]{fontenc}
\usepackage{newtxtext}
\usepackage{microtype}
\usepackage{amsmath,amsthm,mathtools}
\usepackage{newtxmath}
\mathtoolsset{showonlyrefs}
\usepackage{booktabs}
\usepackage{tabularx,array}
\usepackage{enumitem}
\usepackage{float}
\usepackage{xcolor}
\usepackage{hyperref}

\usepackage{multirow}
\usepackage{placeins}
\usepackage{subcaption}

\hypersetup{
  colorlinks=true,
  linkcolor=blue,
  citecolor=blue,
  urlcolor=blue,
  pdftitle={Affine Pricing Models from Group Quantization and Holonomy},
  pdfauthor={Santiago Garcia}
}

\newtheorem{theorem}{Theorem}[section]
\newtheorem{proposition}[theorem]{Proposition}

\newcommand{\Rplus}{\mathbb R_{+}}

\allowdisplaybreaks



\title{\textbf{Affine Pricing Models from Group Quantization and Holonomy}}
\author{Santiago Garc\'ia\\[0.25em]\normalsize Independent Researcher}
\date{September 2026}

\begin{document}

\maketitle

\begin{abstract}

The analytic tractability of affine pricing models is usually expressed through two
complementary formulations: a coordinate-space pricing operator and an
exponential-affine transform representation governed by generalized Riccati
equations. We develop \emph{Affine Holonomy Group Quantization} (AHGQ) as a
geometric framework in which these two formulations arise from the same
underlying structure.

The construction separates the affine pricing symbol into a homogeneous
quadratic sector and a complementary affine sector. The first generates a
finite-dimensional symplectic transport and a centrally extended Lie group,
while the second is represented by a multiplicative holonomy carried by a
thin-path groupoid. Their combination determines an affine
Poincar\'e--Cartan form. Its characteristic dynamics reduce in momentum
variables to the generalized Riccati system and its scalar amplitude, whereas
the coordinate representation recovers the standard affine pricing operator.
Representative Gaussian and square-root models illustrate the construction.
The contribution is structural: AHGQ gives a common geometric origin to the
coordinate and transform representations of continuous-path,
time-homogeneous affine pricing models.

\end{abstract}

\newpage

\tableofcontents
\clearpage
\section{Introduction}
\label{sec:introduction}

Affine models are central in mathematical finance because they combine a rich
description of financial dynamics with a tractable analytical structure.
Classical examples include the Cox--Ingersoll--Ross model of interest rates
\cite{CIR1985} and the Heston stochastic-volatility model
\cite{Heston1993}. The general affine-process framework and its transform
methods are developed further in
\cite{DuffiePanSingleton2000,DuffieFilipovicSchachermayer2003}. A defining
feature of this literature is the exponential-affine transform, whose
coefficients satisfy generalized Riccati equations.

A complementary line of development comes from geometric and group-theoretic
methods in mathematical physics. The Group Approach to Quantization (GAQ) is
related to the geometric formulation of quantization associated with Kostant
and Souriau \cite{Kostant1970,Souriau1997}, while organizing the description
of a dynamical system directly in terms of an underlying Lie group
\cite{Aldaya1982,AldayaWolf1984}. GAQ has previously been applied to
financial models described by quadratic Hamiltonians
\cite{GarciaQuadratic2021}.

A detailed Heston-specific AHGQ construction is given in
\cite{Garcia2026Heston}, including its finite symplectic sector, holonomy,
momentum polarization, Mellin pricing representation, projective Riccati flow,
and numerical validation. The present paper develops the corresponding
general construction for continuous-path, time-homogeneous affine diffusions
with affine drift, affine covariance, and affine discounting or killing. The
class includes Gaussian and square-root models such as Black--Scholes,
Vasicek, CIR, and Heston.

The purpose is not to rederive the generalized Riccati equations or the affine
pricing operator as new financial formulas. Those structures are standard in
affine-process theory. The contribution is instead a geometric correspondence
between them. The affine pricing symbol is decomposed into a homogeneous
quadratic part, represented by a finite-dimensional symplectic Lie-group
sector, and a complementary affine part, represented by multiplicative
thin-path holonomy on a positive central fiber. The two parts combine in one
affine Poincar\'e--Cartan form. Momentum reduction of its characteristic field
produces the generalized Riccati flow and scalar amplitude, whereas the
coordinate representation of the same structure produces the usual affine
pricing operator. Thus the coordinate and transform descriptions are obtained
as complementary representations of a single geometric construction.

This separation also identifies which ingredients depend on strict affinity.
The main text remains within the affine class. Appendix~\ref{app:stock-dependent-default}
uses state-dependent default intensity only as a boundary case: an affine
rate-linked intensity remains inside the construction, whereas a commonly used
inverse-power equity intensity produces discrete momentum translations and
breaks finite-dimensional Riccati closure. That extension is not needed for the
main AHGQ results.

The remainder of the paper is organized as follows.
Section~\ref{sec:quantization-on-group} summarizes the elements of GAQ used in
the construction. Section~\ref{sec:affine-groupoid} introduces the affine
pricing symbol, its symplectic transport, the centrally extended Lie-group
sector, and the thin-path holonomy. Sections~\ref{sec:invariant-generators}
and~\ref{sec:poincare-cartan} develop the invariant generators and the affine
Poincar\'e--Cartan structure. Section~\ref{sec:polarization} derives the
momentum and coordinate representations. Section~\ref{sec:affine-specializations}
collects representative affine models and gives a short CIR specialization.
Section~\ref{sec:conclusion} concludes. Appendix~\ref{app:geometric-action}
discusses the associated phase-space action, while
Appendix~\ref{app:stock-dependent-default} discusses affine and inverse-power
default intensities.


\section{Summary of the Group Approach to Quantization}
\label{sec:quantization-on-group}

In this paper, we use only a few elements of the Group Approach to
Quantization (GAQ) developed by Aldaya and de Azc\'arraga
\cite{Aldaya1982}. The dynamics is encoded in differential
operators derived from an underlying group structure. Two compatible
families of operators arise: one is used to impose constraints that select
the variables retained in a representation, while the other acts on the
resulting functions. 

GAQ also introduces an additional multiplicative variable through a central
extension. In its original quantum-mechanical formulation, GAQ usually
takes \(U(1)\) as the central group, so that this variable is interpreted as
a phase. Here we take \(U=\Rplus\), so that the central coordinate acts
instead as a positive pricing or discounting scale.

A polarization specifies which variables are retained in a representation.
In the affine pricing construction below, the two relevant choices lead to
momentum and coordinate representations: the former produces the affine
transform and its Riccati dynamics, while the latter recovers the usual
pricing operator.

The standard GAQ construction follows a definite sequence. One first
specifies the group law and its central extension. The corresponding left-
and right-invariant generators are then obtained by differentiation. The
central extension determines a distinguished Cartan form, whose
characteristic fields describe the dynamics. A polarization is then
chosen from the left-invariant generators to select a representation, while
the right-invariant generators descend to operators acting in that
representation. Thus the group law, invariant generators, Cartan form, and
polarization are successive parts of a single construction.

AHGQ follows this structure for the finite Lie-group sector generated by
the homogeneous quadratic part of the affine pricing symbol. The complementary
affine sector requires the additional holonomy construction developed
below. The corresponding group structure, invariant generators, Cartan
form, characteristic dynamics, and polarizations are introduced explicitly
in the following sections as they are needed.



\section{Affine Holonomy Group Quantization}
\label{sec:affine-groupoid}

We call the framework developed in this paper \emph{Affine Holonomy Group
Quantization} (AHGQ). It separates the affine pricing symbol into two parts
with different geometric roles. The homogeneous quadratic part \(C_s\)
generates a finite-dimensional linear symplectic transport and the associated
centrally extended Lie-group structure. This is the part of the construction
that follows the standard GAQ framework: the central extension, invariant
generators, and polarizations are all defined within this finite Lie-group
sector.

The complementary part \(C_H\) produces a multiplicative factor accumulated along
paths. This path dependence is described by the thin-path groupoid
\(\mathcal G_A\) and its holonomy \(\mathcal H_H[\gamma]\). This is the
additional ingredient introduced in AHGQ to represent the affine terms that
are not contained in the finite Lie-group sector.

The construction can therefore be followed through four main ingredients.
First, the affine pricing symbol \(C_A\) contains the coefficients of the
pricing model and determines its decomposition into \(C_s\) and \(C_H\).
Second, \(C_s\) generates the linear symplectic transport \(M_s(t)\).
Third, this transport determines the centrally extended Lie group associated
with the symplectic transport and its invariant generators. Finally, the complementary sector \(C_H\) defines
the holonomy carried by the thin-path groupoid. The following sections
develop these ingredients in this order and then show how, together, they
determine the full affine pricing dynamics.

\subsection{Affine symbol \texorpdfstring{\(C_A\)}{C\_A}}
\label{sec:matrix-decomposition}

The affine pricing symbol \(C_A\) provides a phase-space
representation of the affine pricing generator. It collects in a single
function the drift, covariance, and discounting or killing coefficients of the model.
Its dependence on the state variables is affine, while its dependence on the
conjugate variables contains the first- and second-order terms associated
with the pricing operator.

Let \(d\in\mathbb N_{>0}\) and let
\(\mathcal D\subseteq\mathbb R^d\) be the state-variable domain.
We denote the financial state vector by
\(\mathbf x=(x_1,\ldots,x_d)^T\in\mathcal D\), and its conjugate variable by
\(\mathbf p=(p_1,\ldots,p_d)^T\in\mathbb R^d\). In the momentum
representation developed below, \(\mathbf p\) also plays the role of the
transform variable.

The pricing symbol separates naturally into a state-independent part
\(F(\mathbf p)\) and a part linear in the state variables. We write

\begin{equation}
C_A(\mathbf x,\mathbf p)=F(\mathbf p)+\mathbf x^T\mathbf R(\mathbf p), \label{eq:affine-symbol}
\end{equation}

where the state-independent part is

\begin{equation}
F(\mathbf p)=\mathbf b^T\mathbf p+\frac12\mathbf p^T\mathsf A_0\mathbf p-c. \label{eq:F-general}
\end{equation}

Here \(\mathbf b\in\mathbb R^d\) is the constant drift,
\(\mathsf A_0\in\mathbb R^{d\times d}\) is the state-independent covariance
component, and \(c\) is the constant discounting or killing term.

The state-dependent coefficients are collected in the vector-valued function
\(\mathbf R(\mathbf p)\):

\begin{equation}
\mathbf R(\mathbf p)=\mathsf B^T\mathbf p+\frac12\left(\mathbf p^T\mathsf A_1\mathbf p,\ldots,\mathbf p^T\mathsf A_d\mathbf p\right)^T-\mathbf d. \label{eq:R-general}
\end{equation}

We call \(\mathbf R(\mathbf p)\) the \emph{Riccati symbol}
because it becomes the right-hand side of the generalized Riccati equations
in the momentum representation derived in later sections.
Here \(\mathsf B\in\mathbb R^{d\times d}\) is the linear drift matrix,
\(\mathsf A_j\) is the covariance loading associated with the state
component \(x_j\), and \(\mathbf d\) represents state-dependent
discounting or killing.

We assume that the state-independent covariance component and the covariance loadings
are symmetric.
Financial admissibility also requires the full instantaneous covariance
matrix to satisfy

\begin{equation}
\mathsf A(\mathbf x)=\mathsf A_0+\sum_{j=1}^d x_j\mathsf A_j\succeq 0,\qquad \mathbf x\in\mathcal D. \label{eq:affine-covariance}
\end{equation}

This positivity condition is a financial requirement only; the geometric
construction itself does not require it.

For the geometric construction, it is useful to reorganize the same pricing
symbol according to the role played by its terms. We write

\begin{equation}
C_A=C_s+C_H, \label{eq:affine-symbol-decomposition}
\end{equation}

where

\begin{equation}
C_s(\mathbf x,\mathbf p) = \frac12\mathbf p^T\mathsf A_0\mathbf p + \mathbf x^T\mathsf B^T\mathbf p \label{eq:symplectic-sector}
\end{equation}

is the homogeneous quadratic sector generating the symplectic
transport, and

\begin{equation}
C_H(\mathbf x,\mathbf p) = \mathbf b^T\mathbf p-\mathbf x^T\mathbf d-c + \frac12\sum_{j=1}^d x_j\,\mathbf p^T\mathsf A_j\mathbf p \label{eq:holonomy-sector}
\end{equation}

is the complementary sector that will be represented by the affine holonomy.

\subsection{Linear symplectic transport associated with \texorpdfstring{\(C_s\)}{C\_s}}
\label{sec:linear_transport}

The homogeneous quadratic sector of the affine symbol,

\begin{equation}
C_s(\mathbf x,\mathbf p)=\frac12\mathbf p^T\mathsf A_0\mathbf p+\mathbf x^T\mathsf B^T\mathbf p, \label{eq:Cs}
\end{equation}

can be written in symplectic form as

\begin{equation}
C_s=\frac12\mathbf a^TJ_{2d}K_s\mathbf a, \label{eq:cs_symplectic}
\end{equation}

where

\begin{equation}
J_{2d}= \begin{pmatrix} 0&I_d\\ -I_d&0 \end{pmatrix}, \qquad K_s= \begin{pmatrix} -\mathsf B&-\mathsf A_0\\ 0&\mathsf B^T \end{pmatrix}, \qquad \mathbf a=(\mathbf x,\mathbf p)^T. \label{eq:Ks}
\end{equation}

Since \(\mathsf A_0^T=\mathsf A_0\),

\begin{equation}
K_s^TJ_{2d}+J_{2d}K_s=0, \label{eq:Ks-symplectic-condition}
\end{equation}

and therefore the matrix exponential generated by \(K_s\) defines the
symplectic transport

\begin{equation}
M_s(t)=e^{tK_s}\in Sp(2d,\mathbb R),\qquad M_s(t'+t)=M_s(t')M_s(t),\qquad M_s(t)^TJ_{2d}M_s(t)=J_{2d}. \label{eq:Ms-properties}
\end{equation}

By construction, this transport is determined entirely by the homogeneous
quadratic sector \(C_s\); the complementary affine sector \(C_H\) will instead be
represented by the holonomy introduced in
Section~\ref{sec:holonomyaction}.


\subsection{Lie group associated with the symplectic transport}
\label{sec:symplectic-lie-group}

Define the Lie group \(G^s\) as the semidirect product of time translations
with phase-space translations, where the latter are transported by \(M_s(t)\):

\begin{equation}
G^s:=\mathbb R\ltimes_{M_s}\mathbb R^{2d}. \label{eq:symplectic-semidirect-product}
\end{equation}

An element is written \(g=(t,\mathbf a)\), with
\(\mathbf a=(\mathbf x,\mathbf p)^T\). The multiplication is

\begin{equation}
(t',\mathbf a')\star^s(t,\mathbf a) = \bigl(t'+t,M_s(t)\mathbf a'+\mathbf a\bigr). \label{eq:symplectic-group-law}
\end{equation}

Associativity follows from the group property
\(M_s(t'+t)=M_s(t')M_s(t)\).

We use the convention\footnote{This choice is compatible with the
semidirect-product convention of Aldaya,
de Azc\'arraga, and Wolf~\cite{AldayaWolf1984}.}
in which the time coordinate of the second factor transports the
phase-space coordinate of the first factor.

The unit and inverse are

\begin{equation}
\mathbf 1=(0,\mathbf 0),\qquad (t,\mathbf a)^{-1}=\bigl(-t,-M_s(-t)\mathbf a\bigr). \label{eq:symplectic-unit-inverse}
\end{equation}

The group admits the symplectic two-cocycle

\begin{equation}
\epsilon_2^s(g',g) = \frac12\bigl(M_s(t)\mathbf a'\bigr)^TJ_{2d}\mathbf a = \frac12\mathbf a'^TM_s(t)^TJ_{2d}\mathbf a. \label{eq:symplectic-cocycle}
\end{equation}

The group property of \(M_s(t)\), together with its symplecticity, implies
that \(\epsilon_2^s\) is a real normalized group two-cocycle. Hence

\begin{equation}
\widetilde G^s=G^s\times\Rplus \label{eq:extended-symplectic-group}
\end{equation}

is the corresponding central extension, with multiplication

\begin{equation}
(g',\zeta')\widetilde\star^s(g,\zeta) = \Bigl( g'\star^s g,\, \zeta'\zeta\exp\!\bigl(\epsilon_2^s(g',g)\bigr) \Bigr). \label{eq:extended-symplectic-group-law}
\end{equation}


\subsection{Thin-path groupoid and affine holonomy}
\label{sec:holonomyaction}

Groupoids have a long history of use in quantization; in particular,
Hawkins~\cite{Hawkins2008} uses them to unify several geometric and
operator-algebraic constructions. Their role in AHGQ is more limited:
the groupoid organizes the path-dependent holonomy associated with the
complementary affine sector.

Let

\begin{equation}
\mathcal G_A\rightrightarrows G^s \label{eq:thin-path-groupoid}
\end{equation}

denote the thin-path groupoid of \(G^s\); see, for example,
\cite{Meneses2021}.

Its object manifold is \(G^s\), and its arrows are thin-homotopy classes

\begin{equation}
[\gamma]:g_0\longrightarrow g_1,\qquad g_0,g_1\in G^s, \label{eq:thin-path-arrow}
\end{equation}

with source and target maps

\begin{equation}
\mathsf s([\gamma])=g_0,\qquad \mathsf t([\gamma])=g_1. \label{eq:thin-path-source-target}
\end{equation}

The operational rule is simple: two arrows compose when the endpoint of
the first path is the starting point of the second, and composition is path
concatenation. Passing to thin-homotopy classes removes reparametrization and
retracing while retaining the path information used by the holonomy.

The complementary affine sector \(C_H\) defines the multiplicative holonomy

\begin{equation}
\mathcal H_H[\gamma]=\exp\!\left(-\int_\gamma C_H(\mathbf x,\mathbf p)\,dt\right). \label{eq:affine-holonomy}
\end{equation}

For composable paths,

\begin{equation}
\mathcal H_H[\gamma_2\circ\gamma_1]=\mathcal H_H[\gamma_2]\,\mathcal H_H[\gamma_1]. \label{eq:holonomy-composition}
\end{equation}

\begin{proposition}[Thin-path holonomy]
\label{prop:thin-path-holonomy}
Let
\(\alpha_H=C_H(\mathbf x,\mathbf p)\,dt\) on \(G^s\). The map
\(\mathcal H_H[\gamma]=\exp(-\int_\gamma\alpha_H)\) is well defined on
thin-homotopy classes of paths and is multiplicative under path
concatenation.
\end{proposition}

\begin{proof}
The integral of a one-form is invariant under orientation-preserving
reparametrization and changes additively under concatenation. Suppose
\(\gamma_0\) and \(\gamma_1\) are related by a thin homotopy
\(H:[0,1]^2\to G^s\), whose differential has rank at most one. By Stokes'
theorem,
\(
\int_{\gamma_1}\alpha_H-\int_{\gamma_0}\alpha_H
=\int_{[0,1]^2}H^*(d\alpha_H)
\).
Because \(d\alpha_H\) is a two-form and \(dH\) has rank at most one,
\(H^*(d\alpha_H)=0\). Hence the line integral depends only on the thin-homotopy
class. Additivity of the line integral under concatenation gives
\eqref{eq:holonomy-composition} after exponentiation.
\end{proof}

This holonomy acts only on the positive central coordinate of the bundle $\widetilde G^s\longrightarrow G^s$.
For an arrow \([\gamma]:g_0\to g_1\), the path \(\gamma\) connects the two
base points, while the associated holonomy acts on the fiber coordinate
along that path:

\begin{equation}
(g_0,\zeta)\longmapsto \bigl(g_1,\zeta\,\mathcal H_H[\gamma]\bigr). \label{eq:holonomy-fiber-action}
\end{equation}

Thus the groupoid arrow determines the transport from \(g_0\) to \(g_1\),
while \(\mathcal H_H[\gamma]\) gives the corresponding multiplicative
change in the pricing scale. This path composition is separate from the
Lie-group multiplication of \(G^s\).


\subsection{Summary of the AHGQ construction}
\label{sec:ahgq-summary}

The construction developed above consists of two distinct but complementary
parts, which may be summarized schematically as

\begin{equation}
C_s\longrightarrow K_s\longrightarrow M_s(t)\longrightarrow\widetilde G^s, \qquad C_H\longrightarrow\mathcal G_A\longrightarrow\mathcal H_H[\gamma]. \label{eq:ahgq-structure-summary}
\end{equation}

The first chain defines the finite Lie-group sector associated with the
symplectic transport, while the second describes the path-dependent affine
holonomy carried by the thin-path groupoid. These structures are not
combined into a single group law; together they determine the full affine
pricing dynamics.

\section{Invariant generators of the finite Lie-group sector}
\label{sec:invariant-generators}

The invariant vector fields are the infinitesimal form of the composition
law of the centrally extended Lie group \(\widetilde G^s\) associated with
the symplectic transport in Section~\ref{sec:symplectic-lie-group}. In the
GAQ framework, selected left-invariant fields define the constraints that
select a representation, while the right-invariant fields provide operators
acting on the resulting functions.

With the chosen composition convention, the symplectic flow generated by
\(C_s\) enters the left-invariant time generator, while the transport
\(M_s(t)\) appears explicitly in the right-invariant phase-space generators.

Table~\ref{tab:left-generators} gives the left-invariant generators, obtained
by differentiating the group composition law with respect to the unprimed
coordinates and evaluating at the identity. We write

\begin{equation}
\mathbf a=(\mathbf x,\mathbf p)^T,\qquad \nabla_{\mathbf a}=(\nabla_{\mathbf x},\nabla_{\mathbf p})^T,\qquad \Xi=\zeta\partial_\zeta. \label{eq:phase-space-gradient}
\end{equation}

\begin{table}[ht]
\centering
\begin{tabular}{ll}
\hline
Direction & Left-invariant generator of \(\widetilde G^s\) \\[1mm]
\hline
Time
&
\(\displaystyle L_t^s=\partial_t+(K_s\mathbf a)^T\nabla_{\mathbf a}\)
\\[3mm]
Phase space \(\mathbf a\)
&
\(\displaystyle \boldsymbol L_{\mathbf a}^{\,s}
=\nabla_{\mathbf a}-\frac12J_{2d}\mathbf a\,\Xi\)
\\[3mm]
Center
&
\(\displaystyle L_\zeta^s=\Xi\)
\\
\hline
\end{tabular}
\caption{Left-invariant generators of the centrally extended Lie group
\(\widetilde G^s\) associated with the symplectic transport.}
\label{tab:left-generators}
\end{table}

The right-invariant generators are given in
Table~\ref{tab:right-generators}. They are obtained by differentiating the
group composition law in Section~\ref{sec:symplectic-lie-group} with respect
to the primed coordinates and evaluating at the identity.

\begin{table}[ht]
\centering
\begin{tabular}{ll}
\hline
Direction & Right-invariant generator of \(\widetilde G^s\) \\[1mm]
\hline
Time
&
\(\displaystyle R_t^s=\partial_t\)
\\[3mm]
Phase space \(\mathbf a\)
&
\(\displaystyle \boldsymbol R_{\mathbf a}^{\,s}
=M_s(t)^T\left(\nabla_{\mathbf a}+\frac12J_{2d}\mathbf a\,\Xi\right)\)
\\[3mm]
Center
&
\(\displaystyle R_\zeta^s=\Xi\)
\\
\hline
\end{tabular}
\caption{Right-invariant generators of the centrally extended Lie group
\(\widetilde G^s\) associated with the symplectic transport.}
\label{tab:right-generators}
\end{table}


\subsection{Commutators}
\label{sec:commutators}

Let \(L_{a_i}^s\) denote the \(i\)-th component of
\(\boldsymbol L_{\mathbf a}^{\,s}\). The commutators of the left-invariant
generators of the finite Lie-group sector are shown in
Table~\ref{tab:left-commutators}.

\begin{table}[ht]
\centering
\begin{tabular}{ll}
\hline
Commutator & Value \\[1mm]
\hline
\(\displaystyle [L_{a_i}^s,L_{a_j}^s]\)
&
\(\displaystyle (J_{2d})_{ij}\,\Xi\)
\\[2mm]
\(\displaystyle [L_t^s,\boldsymbol L_{\mathbf a}^{\,s}]\)
&
\(\displaystyle -K_s^T\boldsymbol L_{\mathbf a}^{\,s}\)
\\[2mm]
\(\displaystyle [L_t^s,\Xi],\qquad
[\boldsymbol L_{\mathbf a}^{\,s},\Xi]\)
&
\(\displaystyle 0\)
\\
\hline
\end{tabular}
\caption{Commutators of the left-invariant generators of the centrally
extended Lie group \(\widetilde G^s\) associated with the symplectic
transport.}
\label{tab:left-commutators}
\end{table}

The right-invariant generators satisfy the corresponding Lie-algebra
relations with opposite signs. All left-invariant generators commute with
all right-invariant generators.

The holonomy associated with \(C_H\) does not modify the Lie algebra of the
finite Lie-group sector. Its infinitesimal contribution enters separately
through \(C_H\) in the full affine Poincar\'e--Cartan form developed below.

\section{Cartan forms, curvature, and characteristic dynamics}
\label{sec:poincare-cartan}

The affine Poincar\'e--Cartan form combines the contribution of the finite
Lie-group sector generated by \(C_s\) with the holonomy contribution
associated with \(C_H\). Its curvature determines the
characteristic field of the full affine system, which will later yield
the affine momentum dynamics and the generalized Riccati equations. The full
coordinate pricing operator is obtained separately through the coordinate
representation.

We distinguish three related uses of the term \emph{Cartan form}. The
\emph{Maurer--Cartan} form is the canonical Lie-algebra-valued one-form of a Lie
group. In the GAQ construction, we use only its component along the central
generator; this scalar left-invariant one-form is denoted by \(\Theta_s\)
and will be called the \emph{Cartan form} of the finite Lie-group sector. After the
holonomy contribution associated with \(C_H\) is included, the resulting
one-form \(\Theta\) describes the full affine dynamics and will be called
the affine \emph{Poincar\'e--Cartan form}.

\subsection{Cartan form and curvature of the finite Lie-group sector}
\label{sec:finite-cartan-curvature}

The Cartan form \(\Theta_s\) of the finite Lie-group sector is the central
component of the left Maurer--Cartan form of the central extension
\(\widetilde G^s\). Equivalently, it is the left-invariant one-form dual to
the central generator \(\Xi\). It therefore extracts the central component of a
left-invariant vector field. It is normalized on the central direction and
horizontal on the noncentral left-invariant generators:

\begin{equation}
\Theta_s(\Xi)=1,\qquad \Theta_s(\boldsymbol L_{\mathbf a}^{\,s})=\mathbf 0,\qquad \Theta_s(L_t^s)=0. \label{eq:Theta-s-conditions}
\end{equation}

Using the left-invariant generators in Table~\ref{tab:left-generators} and
the expression for \(C_s\) in \eqref{eq:cs_symplectic}, we obtain

\begin{equation}
\Theta_s=\frac{d\zeta}{\zeta}-\frac12\mathbf a^TJ_{2d}\,d\mathbf a+C_s(\mathbf x,\mathbf p)\,dt. \label{eq:Theta-s}
\end{equation}

The corresponding curvature form is

\begin{equation}
\omega_s=d\Theta_s=-\frac12\,d\mathbf a^TJ_{2d}\wedge d\mathbf a+dC_s\wedge dt. \label{eq:omega-s}
\end{equation}

In state--momentum coordinates this becomes

\begin{equation}
\omega_s=-d\mathbf x^T\wedge d\mathbf p+dC_s\wedge dt. \label{eq:omega-s-xp}
\end{equation}

For horizontal vector fields \(X\) and \(Y\),
\(\omega_s(X,Y)=-\Theta_s([X,Y])\). Thus the curvature detects whether
their commutator develops a component in the central direction.

Since

\begin{equation}
dC_s=(\mathsf B^T\mathbf p)^T d\mathbf x+(\mathsf A_0\mathbf p+\mathsf B\mathbf x)^T d\mathbf p, \label{eq:dCs}
\end{equation}

the curvature may also be written explicitly as

\begin{equation}
\omega_s=-d\mathbf x^T\wedge d\mathbf p+\left[(\mathsf B^T\mathbf p)^T d\mathbf x+(\mathsf A_0\mathbf p+\mathsf B\mathbf x)^T d\mathbf p\right]\wedge dt. \label{eq:omega-s-explicit}
\end{equation}

\subsection{Holonomy contribution and affine Poincar\'e--Cartan form}
\label{sec:holonomy-affine-cartan}

The complementary affine sector \(C_H\) acts through the multiplicative holonomy
\(\mathcal H_H[\gamma]\) defined in equation \eqref{eq:affine-holonomy}. For a
fixed-time spatial path, \(dt=0\), and therefore
$\mathcal H_H[\gamma_{\rm sp}]=1$.
Thus the holonomy does not alter the spatial left-invariant generators of the finite Lie-group sector.

For an infinitesimal time displacement,

\begin{equation}
\mathcal H_H=1-C_H(\mathbf x,\mathbf p)\,dt+o(dt), \label{eq:holonomy-infinitesimal}
\end{equation}

Thus its infinitesimal action on the positive central fiber is generated
by the vertical field

\begin{equation}
V_H=-C_H(\mathbf x,\mathbf p)\,\Xi. \label{eq:holonomy-vertical-field}
\end{equation}

This adds a vertical contribution to the time generator of the finite
Lie-group sector. The resulting affine time lift is

\begin{equation}
E_A=L_t^s+V_H=\partial_t+(K_s\mathbf a)^T\nabla_{\mathbf a}-C_H(\mathbf x,\mathbf p)\,\Xi. \label{eq:affine-time-lift}
\end{equation}

The affine Poincar\'e--Cartan form \(\Theta\) is normalized on the central
generator and horizontal on the spatial left-invariant generators and the
affine time lift:

\begin{equation}
\Theta(\Xi)=1,\qquad \Theta(\boldsymbol L_{\mathbf a}^{\,s})=\mathbf 0,\qquad \Theta(E_A)=0. \label{eq:affine-PC-conditions}
\end{equation}

From the defining properties of \(\Theta_s\) and the expression for $V_H$ in equation
\eqref{eq:holonomy-vertical-field},

\begin{equation}
\Theta_s(E_A)=\Theta_s(L_t^s)-C_H\,\Theta_s(\Xi)=-C_H(\mathbf x,\mathbf p). \label{eq:Theta-s-EA}
\end{equation}

Since \(dt(E_A)=1\), horizontality of \(E_A\) is restored by adding
\(C_H\,dt\). Hence

\begin{equation}
\Theta=\Theta_s+C_H(\mathbf x,\mathbf p)\,dt. \label{eq:Theta-from-holonomy}
\end{equation}

Using \(C_A=C_s+C_H\), the full affine Poincar\'e--Cartan form is therefore

\begin{equation}
\Theta=\frac{d\zeta}{\zeta}-\frac12\mathbf a^TJ_{2d}\,d\mathbf a+C_A(\mathbf x,\mathbf p)\,dt =\frac{d\zeta}{\zeta}+\frac12\mathbf p^T d\mathbf x-\frac12\mathbf x^T d\mathbf p+C_A(\mathbf x,\mathbf p)\,dt. \label{eq:affine-PC-form}
\end{equation}

The curvature of the full affine Poincar\'e--Cartan form is

\begin{equation}
\omega=d\Theta=-\frac12\,d\mathbf a^TJ_{2d}\wedge d\mathbf a+dC_A\wedge dt, 
\label{eq:affine-curvature-a}
\end{equation}

or, equivalently,

\begin{equation}
\omega=-d\mathbf x^T\wedge d\mathbf p+dC_A\wedge dt. 
\label{eq:affine-curvature}
\end{equation}

\subsection{Poincar\'e--Cartan characteristic field}
\label{sec:characteristic-direction}

The full Poincar\'e--Cartan form determines the evolution direction of the
affine system. The vector fields that are horizontal for \(\Theta\) and lie
in the kernel of its curvature form the characteristic module

\begin{equation}
\mathcal C_\Theta = \left\{ X\in\mathfrak X(\widetilde G^s): \Theta(X)=0,\qquad d\Theta(X,Y)=0\ \text{for all }Y\in\mathfrak X(\widetilde G^s) \right\}.
 \label{eq:characteristic-module}
\end{equation}

For the affine Poincar\'e--Cartan form in \eqref{eq:affine-PC-form}, the
characteristic module has rank one. Its normalized generator \(X_\Theta\),
defined by \(dt(X_\Theta)=1\), is the distinguished characteristic field:

\begin{equation}
X_\Theta=\partial_t+\mathbf R(\mathbf p)^T\nabla_{\mathbf p}-\bigl(\nabla_{\mathbf p}C_A(\mathbf x,\mathbf p)\bigr)^T\nabla_{\mathbf x}+\eta_A(\mathbf x,\mathbf p)\Xi.
 \label{eq:xtheta}
\end{equation}

Using

\begin{equation}
\nabla_{\mathbf x}C_A=\mathbf R(\mathbf p),\qquad \nabla_{\mathbf p}C_A=\mathbf b+\mathsf B\mathbf x+\mathsf A(\mathbf x)\mathbf p, 
\label{eq:CA-gradients}
\end{equation}

the characteristic field may be written as

\begin{equation}
X_\Theta=X_A+\eta_A(\mathbf x,\mathbf p)\Xi, 
\label{eq:xtheta-decomposition}
\end{equation}

where

\begin{equation}
X_A=\partial_t+\mathbf R(\mathbf p)^T\nabla_{\mathbf p}-\bigl(\mathbf b+\mathsf B\mathbf x+\mathsf A(\mathbf x)\mathbf p\bigr)^T\nabla_{\mathbf x}, 
\label{eq:XA}
\end{equation}

and

\begin{equation}
\begin{aligned}
\eta_A(\mathbf x,\mathbf p)  & = \frac12 \left( \mathbf p^T\nabla_{\mathbf p} + \mathbf x^T\nabla_{\mathbf x}  \right) C_A - C_A  \\
& =  c-\frac12\mathbf b^T\mathbf p+\frac14\mathbf p^T\bigl(\mathsf A(\mathbf x)-\mathsf A_0\bigr)\mathbf p+\frac12\mathbf x^T\mathbf d.
 \label{eq:etaA}
\end{aligned}
\end{equation}

The linear drift matrix \(\mathsf B\) therefore does not contribute to the
central term \(\eta_A\).

\section{Spatial polarization}
\label{sec:polarization}

A spatial polarization selects the representation variables by removing
half of the phase-space variables. In AHGQ it is defined by a maximal
subalgebra of compatible spatial left-invariant generators that is horizontal
for the Cartan form \(\Theta_s\) and closed under commutators.

Since the state generators commute among themselves, as do the momentum
generators,
they define two complementary spatial polarization subalgebras. The momentum
polarization is

\begin{equation}
\mathcal P_{\rm mom}^s = \operatorname{span}_{\mathbb R} \{L_{x_1}^s,\ldots,L_{x_d}^s\}, \label{eq:momentum-spatial-polarization}
\end{equation}

which leads to a representation in the momentum variables \(\mathbf p\), while
the coordinate polarization is

\begin{equation}
\mathcal P_{\rm coord}^s = \operatorname{span}_{\mathbb R} \{L_{p_1}^s,\ldots,L_{p_d}^s\}, \label{eq:coordinate-spatial-polarization}
\end{equation}

which leads to a representation in the state variables \(\mathbf x\).

The relation with the standard GAQ polarization becomes clearer when the
quadratic time generator is included. From

\begin{equation}
[L_t^s,\boldsymbol L_{\mathbf a}^{\,s}] = -K_s^T\boldsymbol L_{\mathbf a}^{\,s}, \label{eq:time-phase-space-commutator}
\end{equation}

one obtains

\begin{equation}
[L_t^s,\boldsymbol L_{\mathbf x}^{\,s}] = \mathsf B^T\boldsymbol L_{\mathbf x}^{\,s} \in \mathcal P_{\rm mom}^s, \qquad [L_t^s,\boldsymbol L_{\mathbf p}^{\,s}] = \mathsf A_0\boldsymbol L_{\mathbf x}^{\,s} -\mathsf B\boldsymbol L_{\mathbf p}^{\,s} \notin \mathcal P_{\rm coord}^s \quad (\mathsf A_0\neq0). \label{eq:time-spatial-polarization-commutators}
\end{equation}

Hence the momentum spatial polarization is invariant under the homogeneous
quadratic evolution, and

\begin{equation}
\widehat{\mathcal P}_{\rm mom}^s = \operatorname{span}_{\mathbb R} \{L_t^s,L_{x_1}^s,\ldots,L_{x_d}^s\} \label{eq:extended-momentum-polarization}
\end{equation}

is a closed first-order polarization of the finite Lie-group sector.

The coordinate spatial polarization is preserved by \(L_t^s\) only when
\(\mathsf A_0=0\), since in that case

\begin{equation}
[L_t^s,\boldsymbol L_{\mathbf p}^{\,s}] = -\mathsf B\boldsymbol L_{\mathbf p}^{\,s} \in \mathcal P_{\rm coord}^s. \label{eq:coordinate-polarization-time-commutator}
\end{equation}

For models with no state-independent covariance component
($\mathsf A_0=0$), including the Heston family and CIR, both the momentum and coordinate spatial
polarizations extend, together with $L_t^s$, to first-order GAQ polarizations of
the finite Lie-group sector. These first-order polarizations describe only the evolution
generated by $C_s$. The remaining affine contribution is carried by $C_H$.


\subsection{Polarized functions}
\label{sec:polarized-functions}

In AHGQ, the representation spaces associated with a spatial polarization
are constructed from functions on the principal bundle
\(\widetilde G^s\to G^s\). 
Let \(\mathcal F=C^\infty(\widetilde G^s,\mathbb C)\) denote the space of
smooth complex-valued functions on \(\widetilde G^s\). For a spatial
polarization subalgebra \(\mathcal P\), the corresponding polarized function
space is

\begin{equation}
\mathcal F_{\mathcal P}=\left\{\Psi\in\mathcal F:X\Psi=0\ \text{for every }X\in\mathcal P,\qquad \Xi\Psi=\Psi\right\}. \label{eq:polarized-function-space}
\end{equation}

The condition \(\Xi\Psi=\Psi\) imposes degree-one homogeneity in the
positive central coordinate \(\zeta\), so that \(\zeta\) is not an
independent physical variable. 

Different choices of \(\mathcal P\) give different representations of the
same finite Lie-group sector. The thin-path groupoid \(\mathcal G_A\) acts
separately through its holonomy, which rescales the positive central fiber.


\subsection{Momentum polarization}
\label{sec:momentum-polarization}

For a momentum-polarized function \(\Psi_{\rm mom} \in \mathcal F_{\mathcal P_{\rm mom}^s}\), the
polarization conditions are

\begin{equation}
L_{x_j}^s\Psi_{\rm mom}=0,\qquad j=1,\ldots,d,\qquad \Xi\Psi_{\rm mom}=\Psi_{\rm mom}. \label{eq:momentum-polarization-conditions}
\end{equation}

Their general solution is

\begin{equation}
\Psi_{\rm mom}=\zeta\exp\!\Bigl(\frac12\mathbf x^T\mathbf p\Bigr)\chi(t,\mathbf p), \label{eq:momentum-polarized-function}
\end{equation}

where \(\chi(t,\mathbf p)\) is the reduced momentum-space function.

The affine characteristic field \(X_\Theta\) preserves the
momentum-polarized function space. Using the polarization conditions
\(L_{x_j}^s\Psi_{\rm mom}=0\) and
\(\Xi\Psi_{\rm mom}=\Psi_{\rm mom}\), together with

\begin{equation}
[X_\Theta,L_{x_j}^s]=\sum_{i=1}^d\frac{\partial R_j}{\partial p_i}L_{x_i}^s, \qquad [X_\Theta,\Xi]=0, \label{eq:characteristic-polarization-commutators}
\end{equation}

we obtain that \(X_\Theta\Psi_{\rm mom} \in \mathcal F_{\mathcal P_{\rm mom}^s}\):

\begin{equation}
\begin{aligned} L_{x_j}^s(X_\Theta\Psi_{\rm mom})  & = X_\Theta(L_{x_j}^s\Psi_{\rm mom}) - [X_\Theta,L_{x_j}^s]\Psi_{\rm mom} =  0,  \\ \Xi(X_\Theta\Psi_{\rm mom}) & = X_\Theta(\Xi\Psi_{\rm mom}) - [X_\Theta,\Xi]\Psi_{\rm mom} = X_\Theta\Psi_{\rm mom}. \end{aligned} \label{eq:characteristic-preserves-momentum-polarization}
\end{equation}


\subsection{Momentum characteristic operator}
\label{sec:momentum-characteristic-operator}

For a momentum-polarized function, the dependence on \(\zeta\) and
\(\mathbf x\) is fixed, so that

\begin{equation}
\Psi_{\rm mom} = \zeta\exp\!\Bigl(\frac12\mathbf x^T\mathbf p\Bigr)\chi(t,\mathbf p). \label{eq:momentum-polarized-function-repeated}
\end{equation}

We therefore define the momentum reduction map by

\begin{equation}
\mathcal R_{\rm mom}\Psi=\zeta^{-1}\exp\!\Bigl(-\frac12\mathbf x^T\mathbf p\Bigr)\Psi,\qquad \mathcal R_{\rm mom}\Psi_{\rm mom}=\chi(t,\mathbf p). \label{eq:momentum-reduction-map}
\end{equation}

We call $\chi(t,\mathbf p)$ a reduced momentum polarized function.  Since \(X_\Theta\) preserves \(\mathcal F_{\mathcal P_{\rm mom}^s}\), its momentum
realization is defined by

\begin{equation}
X_\Theta^{\rm mom} \chi(t,\mathbf p)=\mathcal R_{\rm mom}\!\left(X_\Theta\Psi_{\rm mom}\right). \label{eq:momentum-reduced-action}
\end{equation}

Substitution of \eqref{eq:momentum-polarized-function} gives

\begin{equation}
X_\Theta^{\rm mom}\chi(t,\mathbf p)=\Bigl[\partial_t+\mathbf R(\mathbf p)^T\nabla_{\mathbf p}+
\frac12\mathbf x^T\mathbf R(\mathbf p)-\frac12\mathbf p^T\nabla_{\mathbf p}C_A+\eta_A\Bigr]\chi(t,\mathbf p). \label{eq:momentum-reduced-intermediate}
\end{equation}

Using

\begin{equation}
\frac12\mathbf x^T\mathbf R(\mathbf p)-\frac12\mathbf p^T\nabla_{\mathbf p}C_A+\eta_A=-F(\mathbf p), \label{eq:momentum-reduction-identity}
\end{equation}

the dependence on \(\mathbf x\) cancels, and the momentum characteristic
operator reduces to

\begin{equation}
X_\Theta^{\rm mom}=\partial_t+\mathbf R(\mathbf p)^T\nabla_{\mathbf p}-F(\mathbf p). \label{eq:xmom}
\end{equation}

\begin{proposition}[Momentum reduction and Riccati flow]
\label{prop:momentum-riccati}
The affine characteristic field \(X_\Theta\) preserves the momentum-polarized
space \(\mathcal F_{\mathcal P_{\rm mom}^s}\). On reduced functions its action
is the first-order operator \eqref{eq:xmom}. Consequently its characteristics
satisfy
\(\dot{\mathbf p}=\mathbf R(\mathbf p)\), while the scalar amplitude satisfies
\(\dot\chi=F(\mathbf p)\chi\). The momentum trajectory is therefore the
generalized Riccati flow associated with the affine symbol
\eqref{eq:affine-symbol}.
\end{proposition}

\begin{proof}
Preservation follows from the commutators
\eqref{eq:characteristic-polarization-commutators} and the polarization
conditions. Applying the reduction map \eqref{eq:momentum-reduction-map} to
\(X_\Theta\Psi_{\rm mom}\) gives
\eqref{eq:momentum-reduced-intermediate}; the identity
\eqref{eq:momentum-reduction-identity} removes the state dependence and yields
\eqref{eq:xmom}. The characteristic equations of this reduced first-order
operator are the stated equations.
\end{proof}
\subsection{European claims in the momentum representation}
\label{sec:european-momentum-representation}

In the momentum representation, each initial momentum \(\mathbf u\) labels an
exponential mode \(\exp(\mathbf u^T\mathbf x)\). The reduced polarized function
\(\chi(t,\mathbf p)\) describes the amplitude generated along the
corresponding momentum characteristic and is normalized at expiry by
\(\chi(0,\mathbf u)=1\).

If a European payoff can be represented as a superposition of these exponential
modes, the momentum
propagator evolves each mode away from expiry, and the coordinate-space
price is obtained by superposing the propagated modes with the same payoff
coefficients.

\subsubsection{Momentum flow and propagator}

\label{sec:flows_momentum_propagator}

In this section, we use reverse time, with \(t=0\) corresponding to expiry.

To construct the momentum propagator, we label the initial momentum by \(\mathbf u\) and normalize the amplitude at the payoff time \(t=0\) by \(\chi(0,\mathbf u)=1\).
The subsequent evolution is determined by the \(X_\Theta^{\rm mom}\chi=0\) flows.

\begin{equation}
\frac{d\mathbf p}{dt}=\mathbf R(\mathbf p),\qquad \frac{d}{dt}\chi(t,\mathbf p(t))=F(\mathbf p(t))\chi(t,\mathbf p(t)),\qquad \mathbf p(0)=\mathbf u,\qquad \chi(0,\mathbf u)=1. 
\label{eq:momentum-characteristic-system}
\end{equation}

Writing \(\mathbf p(t)=\psi(t,\mathbf u)\), the momentum trajectory satisfies the generalized Riccati system of the affine model.

\begin{equation}
\partial_t\psi(t,\mathbf u)=\mathbf R\!\bigl(\psi(t,\mathbf u)\bigr),\qquad \psi(0,\mathbf u)=\mathbf u. 
\label{eq:generalized-riccati}
\end{equation}

The corresponding scalar amplitude is described by

\begin{equation}
\partial_t\phi(t,\mathbf u)=F\!\bigl(\psi(t,\mathbf u)\bigr),\qquad \phi(0,\mathbf u)=0,\qquad  \longrightarrow \quad
\phi(t,\mathbf u)=\int_0^t F\!\bigl(\psi(s,\mathbf u)\bigr)\,ds.
 \label{eq:phi-characteristic}
\end{equation}

The resulting scalar amplitude is a multiplicative cocycle on the Riccati
characteristic-flow groupoid. Appendix~\ref{app:cocycle-relations}
clarifies its relation to the affine thin-path holonomy and the
quadratic central extension.

At expiry, the initial momentum \(\mathbf u\) labels the exponential
transform component \(\exp(\mathbf u^T\mathbf x)\) appearing in the
representation of the payoff. Under the momentum flow,
\(\mathbf u\) evolves to \(\psi(t,\mathbf u)\), while the
corresponding amplitude is \(\exp\!\bigl(\phi(t,\mathbf u)\bigr)\).
The momentum propagator is therefore

\begin{equation}
K(t;\mathbf x;\mathbf u)
=
\exp\!\left(
\phi(t,\mathbf u)
+
\left\langle \psi(t,\mathbf u),\mathbf x\right\rangle
\right).
\label{eq:momentum-propagator}
\end{equation}

where \(\phi(t,\mathbf u)\) is scalar-valued and \(\psi(t,\mathbf u)\) is the transported momentum vector.

Thus \(X_\Theta^{\rm mom}\) determines the affine transform propagator: the momentum flow gives the generalized Riccati system for \(\psi\), while the scalar amplitude equation determines \(\phi\). When a component of the Riccati system reduces to a scalar equation, as in Heston or CIR, it may be integrated through a projective \(SL(2,\mathbb C)\) flow.

Finally, let \(\widehat f(\mathbf u)\) denote the transform coefficient of the payoff, so that the payoff in coordinate space is

\begin{equation}
f(\mathbf x)= \int_\Gamma \widehat f(\mathbf u)\exp(\mathbf u^T\mathbf x)\,d\mathbf u.
\label{eq:payoff-transform}
\end{equation}

The contour and normalization depend on the transform being used.  The European claim price is obtained by superposing the propagated transform
components:

\begin{equation}
V(t,\mathbf x)=\int_\Gamma K(t,\mathbf x;\mathbf u)\widehat f(\mathbf u)\,d\mathbf u.
\label{eq:affine-transform-price}
\end{equation}

\subsection{Coordinate polarization and pricing operator}
\label{sec:coordinate-representation}

The coordinate-polarized function space
\(\mathcal F_{\rm coord}:=\mathcal F_{\mathcal P_{\rm coord}^s}\) is, from
Section~\ref{sec:polarization},

\begin{equation}
\mathcal F_{\rm coord}=\left\{\Psi\in\mathcal F:L_{p_j}^s\Psi=0,\ j=1,\ldots,d,\qquad \Xi\Psi=\Psi\right\}.
\label{eq:polarization_coord}
\end{equation}

The conditions in \eqref{eq:polarization_coord} give

\begin{equation}
\Psi_{\rm coord}=\zeta\exp\!\Bigl(-\frac12\mathbf x^T\mathbf p\Bigr)V(t,\mathbf x),\qquad \Psi_{\rm coord}\in\mathcal F_{\rm coord}.
\label{eq:coordinate-polarized-function}
\end{equation}

For a coordinate-polarized function, the dependence on \(\zeta\) and
\(\mathbf p\) is fixed by the polarization conditions and can therefore be
factored out. The remaining function \(V(t,\mathbf x)\) is the reduced
coordinate representation.

When a first-order polarization is not sufficient to impose the required
functional constraints, one may use a higher-order polarization. A
higher-order polarization is a maximal subalgebra of the left-invariant
enveloping algebra containing higher-order differential operators, compatible
with the first-order polarization and with the central equivariance
condition. Unlike a first-order polarization, it need not define a foliation
of the group manifold \cite{GarciaQuadratic2021}.

Modulo the coordinate-polarization constraints, and normalizing the
coefficient of \(L_t^s\) to one, a general second-order parabolic operator
relevant to the coordinate representation may be written as

\begin{equation}
\mathcal D_L
=
L_t^s
+
\boldsymbol\alpha(\mathbf x,\mathbf p)^T\boldsymbol L_{\mathbf x}^{\,s}
+
\frac12\sum_{i,j=1}^d
\mathsf G_{ij}(\mathbf x,\mathbf p)L_{x_i}^sL_{x_j}^s
+
\rho(\mathbf x,\mathbf p)\Xi,
\qquad
\mathsf G^T=\mathsf G.
\label{eq:general-left-enveloping-operator}
\end{equation}

For the affine symbol \(C_A=C_s+C_H\), requiring the reduced operator to be
independent of \(\mathbf p\) and to reproduce the affine drift, covariance,
and killing terms determines the coefficients in
\eqref{eq:general-left-enveloping-operator}. Straightforward algebra gives the
affine higher-order operator

\begin{equation}
\mathcal D_A^L
=
L_t^s
-
\left(\nabla_{\mathbf p}C_H\right)^T\boldsymbol L_{\mathbf x}^{\,s}
-
\frac12\sum_{i,j=1}^d
\mathsf A_{ij}(\mathbf x)L_{x_i}^sL_{x_j}^s
-
C_H\Xi.
\label{eq:affine-left-higher-order-operator}
\end{equation}

Using the affine time lift
\(E_A=L_t^s-C_H\Xi\) from
Section~\ref{sec:holonomy-affine-cartan}, this may equivalently be written as

\begin{equation}
\mathcal D_A^L
=
E_A
-
\left(\nabla_{\mathbf p}C_H\right)^T\boldsymbol L_{\mathbf x}^{\,s}
-
\frac12\sum_{i,j=1}^d
\mathsf A_{ij}(\mathbf x)L_{x_i}^sL_{x_j}^s.
\label{eq:affine-left-higher-order-operator-EA}
\end{equation}

On the coordinate-polarized space, the same straightforward reduction gives

\begin{equation}
\mathcal D_A^L\Psi_{\rm coord}
=
\zeta e^{-\frac12\mathbf x^T\mathbf p}
\left(\partial_t-\mathcal L_A\right)V(t,\mathbf x),
\label{eq:affine-left-coordinate-action}
\end{equation}

where

\begin{equation}
\mathcal L_A
=
(\mathbf b+\mathsf B\mathbf x)^T\nabla_{\mathbf x}
+\frac12\operatorname{Tr}\!\left[\mathsf A(\mathbf x)\nabla_{\mathbf x}^2\right]
-\bigl(c+\mathbf d^T\mathbf x\bigr).
\label{eq:affine-coordinate-generator}
\end{equation}

Therefore the higher-order polarization condition

\begin{equation}
\mathcal D_A^L\Psi_{\rm coord}=0
\label{eq:coordinate-higher-order-polarization}
\end{equation}

is equivalent to the affine pricing equation

\begin{equation}
(\partial_t-\mathcal L_A)V(t,\mathbf x)=0.
\label{eq:coordinate-pricing-equation}
\end{equation}

\subsubsection{Reduction of the affine symbol}
\label{sec:reduction-affine-symbol}

The preceding construction may be expressed more compactly by reducing the
affine symbol itself. Define the coordinate reduction map

\begin{equation}
\mathcal R_{\rm coord}\Psi
=
\zeta^{-1}\exp\!\Bigl(\frac12\mathbf x^T\mathbf p\Bigr)\Psi,
\qquad
\mathcal R_{\rm coord}\Psi_{\rm coord}=V(t,\mathbf x).
\label{eq:coordinate-reduction-map}
\end{equation}

On the reduced coordinate representation, the canonical phase-space variables
are represented by

\begin{equation}
\mathbf P:=\nabla_{\mathbf x},
\qquad
\mathbf X:=\mathbf x,
\qquad
[\mathbf P,\mathbf X^T]=I_d.
\label{eq:PX-coordinate-reduction}
\end{equation}

Thus \(\mathbf P\) and \(\mathbf X\) are the reduced operator representatives
of the phase-space variables \(\mathbf p\) and \(\mathbf x\) entering the
affine symbol \(C_A(\mathbf x,\mathbf p)\). We therefore define the ordered
coordinate realization of the affine symbol by

\begin{equation}
\widehat C_A^{\,\rm coord}
:=
F(\mathbf P)+\mathbf X^T\mathbf R(\mathbf P),
\label{eq:affine-coordinate-symbol-operator}
\end{equation}

where the ordering follows the displayed affine form
\(C_A(\mathbf x,\mathbf p)=F(\mathbf p)+\mathbf x^T\mathbf R(\mathbf p)\).
Its action on the reduced coordinate function is

\begin{equation}
\widehat C_A^{\,\rm coord}V(t,\mathbf x)
=
\left[
F(\nabla_{\mathbf x})
+
\mathbf x^T\mathbf R(\nabla_{\mathbf x})
\right]V(t,\mathbf x)
=
\mathcal L_A V(t,\mathbf x).
\label{eq:affine-coordinate-symbol-action}
\end{equation}

Hence the reduction of the affine symbol gives directly

\begin{equation}
\mathcal L_A
=(\mathbf b+\mathsf B\mathbf x)^T\nabla_{\mathbf x}
+\frac12\operatorname{Tr}\!\left[\mathsf A(\mathbf x)\nabla_{\mathbf x}^2\right]
-\bigl(c+\mathbf d^T\mathbf x\bigr),
\label{eq:affine-coordinate-generator-repeated}
\end{equation}

in agreement with \eqref{eq:affine-coordinate-generator}. The left-enveloping
construction above therefore determines the higher-order polarization
operator, while the reduced affine symbol provides the corresponding compact
coordinate-space realization.

As an additional check, we verify that the reduced operator \(\mathcal L_A\)
is consistent with the momentum representation. For brevity, in the formulas
below we suppress the arguments and write
\(\psi=\psi(t,\mathbf u)\),
\(\phi=\phi(t,\mathbf u)\), and \(K=K(t,\mathbf x;\mathbf u)\).

The momentum propagator \(K\) defined in
\eqref{eq:momentum-propagator} satisfies

\begin{equation}
\partial_tK
=
\left[F(\psi)+\mathbf x^T\mathbf R(\psi)\right]K
=
C_A(\mathbf x,\psi)K.
\label{eq:momentum-propagator-time}
\end{equation}

Here \(\partial_t\psi=\mathbf R(\psi)\) and
\(\partial_t\phi=F(\psi)\). Using
\(\nabla_{\mathbf x}K=\psi K\) and
\(\nabla_{\mathbf x}^2K=\psi\psi^TK\), the coordinate operator gives

\begin{equation}
\mathcal L_AK
=
\left[
(\mathbf b+\mathsf B\mathbf x)^T\psi
+\frac12\psi^T\mathsf A(\mathbf x)\psi
-\bigl(c+\mathbf d^T\mathbf x\bigr)
\right]K
=
C_A(\mathbf x,\psi)K.
\label{eq:coordinate-propagator-action}
\end{equation}

Together with \eqref{eq:momentum-propagator-time}, this yields

\begin{equation}
(\partial_t-\mathcal L_A)K=0.
\label{eq:momentum-coordinate-consistency}
\end{equation}

\begin{proposition}[Coordinate reduction and consistency]
\label{prop:coordinate-operator}
The higher-order operator \(\mathcal D_A^L\) in
\eqref{eq:affine-left-higher-order-operator} preserves the
coordinate-polarized space and reduces to
\(\partial_t-\mathcal L_A\). Equivalently, the ordered coordinate realization
\(F(\mathbf P)+\mathbf X^T\mathbf R(\mathbf P)\) of the affine symbol reduces
to the affine pricing operator \(\mathcal L_A\) in
\eqref{eq:affine-coordinate-generator}. Moreover, the exponential-affine
propagator obtained from Proposition~\ref{prop:momentum-riccati} satisfies the
coordinate evolution equation \eqref{eq:momentum-coordinate-consistency}.
\end{proposition}

\begin{proof}
The reduction \eqref{eq:affine-left-coordinate-action} follows directly from
\eqref{eq:general-left-enveloping-operator} and
\eqref{eq:affine-left-higher-order-operator}. The symbol reduction
\eqref{eq:affine-coordinate-symbol-action} gives the same operator
\(\mathcal L_A\). Finally, \eqref{eq:momentum-propagator-time} and
\eqref{eq:coordinate-propagator-action} give the same factor
\(C_A(\mathbf x,\psi)\), which proves
\eqref{eq:momentum-coordinate-consistency}.
\end{proof}

Since \(K(0,\mathbf x;\mathbf u)=\exp(\mathbf u^T\mathbf x)\), the momentum
propagator can be characterized as the coordinate evolution of the exponential
transform component labeled by \(\mathbf u\).


\section{Representative affine models}
\label{sec:affine-specializations}

The purpose of this section is to display how the general AHGQ construction
specializes across representative affine models. The tables make explicit the
affine data, the decomposition into symplectic and holonomy sectors, and the
coordinate pricing operators obtained from the general formulas of the
preceding sections.

We use standard financial notation: \(\sigma\) denotes volatility,
\(\rho\) correlation, \(r\) the risk-free rate, \(\delta\) the dividend
yield, \(\kappa\) the mean-reversion speed, \(\theta\) the long-run mean,
\(v\) instantaneous variance, and \(\sigma_\nu\) the volatility of variance.
In CIR, the state variable is the short rate; in all the other models 
\(x\) represents the logarithm of the stock price $S$.

\subsection{Affine data and geometric decomposition}
\label{sec:representative-affine-models}

The tables below collect the affine data for several representative models.

Gaussian models such as Black--Scholes and Vasicek also fall within the general AHGQ construction. 
For Black--Scholes the affine symbol contains no state-dependent
covariance contribution. Vasicek is a Gaussian
mean-reverting affine model, with linear drift and state-independent
covariance. In both cases the momentum equations reduce to linear or
elementary Gaussian flows. CIR and Heston display
the state-dependent covariance and Riccati structure central to the affine
construction.

We first give the covariance loadings \(\mathsf A_j\), followed by
the functions \(F\) and \(\mathbf R\) defining
\(C_A=F+\mathbf x^T\mathbf R\). We then give the corresponding affine symbol
decomposition \(C_A=C_s+C_H\).

\begin{table}[H]
\centering
\small
\setlength{\tabcolsep}{4pt}
\renewcommand{\arraystretch}{1.5}
\begin{tabular}{
>{\raggedright\arraybackslash}m{1.8cm}
>{\raggedright\arraybackslash}m{3cm}
>{\centering\arraybackslash}m{2.5cm}
>{\centering\arraybackslash}m{8.0cm}
}
\toprule
Model & Phase-space variables & State-independent covariance \(\mathsf A_0\) & State-dependent covariance matrices \(\mathsf A_j\) \\
\midrule
Heston &
\(\mathbf x=(x,v)^T\)

\(\mathbf p=(p_x,p_v)^T\) &
\(\mathsf A_0=0_{2\times2}\) &
\(\displaystyle \mathsf A_1=0_{2\times2},\qquad \mathsf A_2=\begin{pmatrix}1&\rho\sigma_\nu\\ \rho\sigma_\nu&\sigma_\nu^2\end{pmatrix}\) \\
\midrule
Two-factor Heston &
\(\mathbf x=(x,v_1,v_2)^T\)

\(\mathbf p=(p_x,p_{v_1},p_{v_2})^T\) &
\(\mathsf A_0=0_{3\times3}\) &
\(\displaystyle \begin{aligned}\mathsf A_1&=0_{3\times3},\\[2mm]\mathsf A_2&=\begin{pmatrix}1&\rho_1\sigma_1&0\\ \rho_1\sigma_1&\sigma_1^2&0\\ 0&0&0\end{pmatrix},\\[2mm]\mathsf A_3&=\begin{pmatrix}1&0&\rho_2\sigma_2\\ 0&0&0\\ \rho_2\sigma_2&0&\sigma_2^2\end{pmatrix}\end{aligned}\) \\
\midrule
CIR &
\(\mathbf x=x\)

\(\mathbf p=p\) &
\(\mathsf A_0=0\) &
\(\mathsf A_1=\sigma^2\) \\
\midrule
Black--Scholes &
\(\mathbf x=x=\log S\)

\(\mathbf p=p\) &
\(\mathsf A_0=\sigma^2\) &
\(\mathsf A_1=0\) \\
\midrule
Vasicek &
\(\mathbf x=x\)

\(\mathbf p=p\) &
\(\mathsf A_0=\sigma^2\) &
\(\mathsf A_1=0\) \\
\bottomrule
\end{tabular}
\caption{Constant and state-dependent covariance matrices in
\(\mathsf A(\mathbf x)=\mathsf A_0+\sum_jx_j\mathsf A_j\)
for representative affine models.}
\label{tab:affine-covariance-matrices}
\end{table}

\begin{table}[H]
\centering
\small
\setlength{\tabcolsep}{4pt}
\renewcommand{\arraystretch}{1.4}
\begin{tabular*}{\textwidth}{
@{\extracolsep{\fill}}
>{\raggedright\arraybackslash}m{2.1cm}
>{\raggedright\arraybackslash}m{5.1cm}
>{\raggedright\arraybackslash}m{7.7cm}
}
\toprule
Model & State-independent part \(F(\mathbf p)\) & Riccati symbol \(\mathbf R(\mathbf p)\) \\
\midrule
Heston & \(\displaystyle F(\mathbf p)=(r-\delta)p_x+\kappa\theta p_v-r\) & \(\displaystyle \mathbf R(\mathbf p)=\left(0,\frac12(p_x^2-p_x)+(\rho\sigma_\nu p_x-\kappa)p_v+\frac12\sigma_\nu^2p_v^2\right)^T\) \\
\midrule
Two-factor Heston & \(\displaystyle F(\mathbf p)=(r-\delta)p_x+\sum_{i=1}^2\kappa_i\theta_i p_{v_i}-r\) & \(\displaystyle \begin{aligned}\mathbf R(\mathbf p)&=(0,R_1(\mathbf p),R_2(\mathbf p))^T,\\ R_i(\mathbf p)&=\frac12(p_x^2-p_x)+(\rho_i\sigma_i p_x-\kappa_i)p_{v_i}+\frac12\sigma_i^2p_{v_i}^2\end{aligned}\) \\
\midrule
CIR & \(\displaystyle F(p)=\kappa\theta p\) & \(\displaystyle R(p)=\frac12\sigma^2p^2-\kappa p-1\) \\
\midrule
Black--Scholes & \(\displaystyle F(p)=\frac12\sigma^2p^2+\left(r-\delta-\frac12\sigma^2\right)p-r\) & \(\displaystyle R(p)=0\) \\
\midrule
Vasicek & \(\displaystyle F(p)=\frac12\sigma^2p^2+\kappa\theta p\) & \(\displaystyle R(p)=-\kappa p-1\) \\
\bottomrule
\end{tabular*}
\caption{State-independent part \(F\) and Riccati symbol \(\mathbf R\) for
representative affine models. Together they reconstruct the affine pricing
symbol through
\(C_A(\mathbf x,\mathbf p)=F(\mathbf p)+\mathbf x^T\mathbf R(\mathbf p)\).
For the two-factor Heston model,
\(\mathbf R=(0, R_1(\mathbf p),R_2(\mathbf p))^T\),  hence the log-price momentum  \(p_x\)  remains
constant, while the two variance momenta satisfy independent scalar
Riccati equations conditional on \(p_x\).  }
\label{tab:affine-F-R}
\end{table}

\begin{table}[H]
\centering
\small
\setlength{\tabcolsep}{4pt}
\renewcommand{\arraystretch}{1.4}
\begin{tabular*}{\textwidth}{
@{\extracolsep{\fill}}
>{\raggedright\arraybackslash}m{2.2cm}
>{\raggedright\arraybackslash}m{4.4cm}
>{\raggedright\arraybackslash}m{7.0cm}
}
\toprule
Model & Symplectic sector \(C_s\) & Holonomy sector \(C_H\) \\
\midrule
Heston & \(\displaystyle C_s=-v\left(\frac12p_x+\kappa p_v\right)\) & \(\displaystyle C_H=(r-\delta)p_x+\kappa\theta p_v-r+\frac12v\left(p_x^2+2\rho\sigma_\nu p_xp_v+\sigma_\nu^2p_v^2\right)\) \\
\midrule
Two-factor Heston & \(\displaystyle C_s=-\sum_{i=1}^2v_i\left(\frac12p_x+\kappa_i p_{v_i}\right)\) & \(\displaystyle C_H=(r-\delta)p_x+\sum_{i=1}^2\kappa_i\theta_i p_{v_i}-r+\frac12\sum_{i=1}^2v_i\left(p_x^2+2\rho_i\sigma_i p_xp_{v_i}+\sigma_i^2p_{v_i}^2\right)\) \\
\midrule
CIR & \(\displaystyle C_s=-\kappa xp\) & \(\displaystyle C_H=\kappa\theta p-x+\frac12\sigma^2xp^2\) \\
\midrule
Black--Scholes & \(\displaystyle C_s=\frac12\sigma^2p^2\) & \(\displaystyle C_H=\left(r-\delta-\frac12\sigma^2\right)p-r\) \\
\midrule
Vasicek & \(\displaystyle C_s=\frac12\sigma^2p^2-\kappa xp\) & \(\displaystyle C_H=\kappa\theta p-x\) \\
\bottomrule
\end{tabular*}
\caption{Decomposition \(C_A=C_s+C_H\) for representative affine models.
The homogeneous quadratic sector \(C_s\) determines the finite symplectic
transport, while \(C_H\) is represented by the holonomy.}
\label{tab:affine-symbol-decomposition}
\end{table}
\subsection{Coordinate pricing operators}
\label{sec:representative-pricing-operators}

Substitution of the affine data above into the general coordinate operator
\eqref{eq:affine-coordinate-generator} gives the pricing operators in
Table~\ref{tab:representative-pricing-operators}. These operators are the
coordinate-space counterparts of the momentum-space data summarized by
\(F\) and \(\mathbf R\) in Table~\ref{tab:affine-F-R}.

\begin{table}[H]
\centering
\scriptsize
\setlength{\tabcolsep}{4pt}
\renewcommand{\arraystretch}{1.45}
\begin{tabularx}{\textwidth}{
>{\raggedright\arraybackslash}p{2.2cm}
>{\raggedright\arraybackslash}X}
\toprule
Model & Coordinate pricing operator \\
\midrule
Heston &
\(\displaystyle
\mathcal L_H=
\frac12v\,\partial_x^2
+\rho\sigma_\nu v\,\partial_x\partial_v
+\frac12\sigma_\nu^2v\,\partial_v^2
+\left(r-\delta-\frac12v\right)\partial_x
+\kappa(\theta-v)\partial_v-r
\) \\
\midrule
Two-factor Heston &
\(\displaystyle
\begin{aligned}
\mathcal L_{2H}={}&
\frac12(v_1+v_2)\partial_x^2
+\sum_{i=1}^2\rho_i\sigma_i v_i\,\partial_x\partial_{v_i}
+\frac12\sum_{i=1}^2\sigma_i^2v_i\,\partial_{v_i}^2 \\
&+\left(r-\delta-\frac12(v_1+v_2)\right)\partial_x
+\sum_{i=1}^2\kappa_i(\theta_i-v_i)\partial_{v_i}-r
\end{aligned}
\) \\
\midrule
CIR &
\(\displaystyle
\mathcal L_{\mathrm{CIR}}
=\frac12\sigma^2x\,\partial_x^2
+\kappa(\theta-x)\partial_x-x
\) \\
\midrule
Black--Scholes &
\(\displaystyle
\mathcal L_{\mathrm{BS}}
=\frac12\sigma^2\partial_x^2
+\left(r-\delta-\frac12\sigma^2\right)\partial_x-r,
\qquad x=\log S
\) \\
\midrule
Vasicek &
\(\displaystyle
\mathcal L_{\mathrm{V}}
=\frac12\sigma^2\partial_x^2
+\kappa(\theta-x)\partial_x-x
\) \\
\bottomrule
\end{tabularx}
\caption{Coordinate pricing operators obtained from the general AHGQ
coordinate representation. In the CIR and Vasicek rows, \(x\) denotes the
short rate.}
\label{tab:representative-pricing-operators}
\end{table}

A detailed application of AHGQ to the Heston model is given in
\cite{Garcia2026Heston}. That work develops the model-specific symplectic and
holonomy construction, the momentum polarization and Mellin representation,
the projective linearization of the Heston Riccati equation, numerical
validation against the standard Heston solution, and the Black--Scholes
reduction. Those Heston-specific derivations are therefore not repeated here.

\subsection{CIR as a one-dimensional specialization}
\label{sec:cir-specialization}

The CIR short-rate model gives a compact non-Heston illustration of the whole
construction. With short rate \(x\geq0\),

\begin{equation}
dx_t=\kappa(\theta-x_t)\,dt+\sigma\sqrt{x_t}\,dW_t,
\label{eq:cir-sde-specialization}
\end{equation}

and discounting at the short rate, the affine pricing symbol is

\begin{equation}
C_A(x,p)=\kappa\theta p+x\left(\frac12\sigma^2p^2-\kappa p-1\right).
\label{eq:cir-symbol-specialization}
\end{equation}

The AHGQ decomposition is

\begin{equation}
C_s(x,p)=-\kappa xp,
\qquad
C_H(x,p)=\kappa\theta p-x+\frac12\sigma^2xp^2.
\label{eq:cir-split-specialization}
\end{equation}

Thus the finite symplectic sector is generated by the homogeneous term
\(-\kappa xp\), while the constant drift, short-rate discounting, and
state-dependent covariance are carried by \(C_H\). The reduced momentum
characteristics are

\begin{equation}
\dot p=\frac12\sigma^2p^2-\kappa p-1,
\qquad
\dot\phi=\kappa\theta p,
\label{eq:cir-riccati-specialization}
\end{equation}

which are the standard CIR Riccati and amplitude equations. For a unit payoff,
\(p(0)=0\), and the momentum propagator takes the exponential-affine form

\begin{equation}
K(t;x;0)=\exp\!\bigl(\phi(t,0)+\psi(t,0)x\bigr).
\label{eq:cir-propagator-specialization}
\end{equation}

The coordinate representation gives

\begin{equation}
\mathcal L_{\mathrm{CIR}}
=\frac12\sigma^2x\,\partial_x^2
+\kappa(\theta-x)\partial_x-x,
\label{eq:cir-coordinate-specialization}
\end{equation}

and Proposition~\ref{prop:coordinate-operator} implies
\((\partial_t-\mathcal L_{\mathrm{CIR}})K=0\). The CIR example therefore
shows in the simplest state-dependent-covariance setting how the same affine
symbol produces the symplectic/holonomy decomposition, the Riccati transform,
and the coordinate pricing operator.

\section{Conclusion}
\label{sec:conclusion}

We have developed Affine Holonomy Group Quantization for continuous-path,
time-homogeneous affine pricing models. The construction separates the
homogeneous quadratic sector of the affine pricing symbol from its
complementary affine sector. The first generates a finite-dimensional
symplectic transport and its centrally extended Lie group; the second defines
a multiplicative holonomy on the positive central fiber of a thin-path
groupoid. Their combination gives the affine Poincar\'e--Cartan form and its
characteristic dynamics.

The main result is a geometric correspondence between the two standard
representations of affine pricing. Proposition~\ref{prop:momentum-riccati}
shows that momentum reduction of the characteristic field gives the generalized
Riccati system and scalar affine-transform amplitude. Proposition~\ref{prop:coordinate-operator}
shows that the coordinate representation recovers the standard affine pricing
operator, and that the exponential-affine propagator obtained from the momentum
flow solves the corresponding coordinate evolution equation. The construction
therefore does not replace the usual affine-process formulas; it organizes them
as complementary reductions of one geometric structure.

The representative models in Section~\ref{sec:affine-specializations} show how
the same construction applies to Gaussian and square-root diffusions. The CIR
specialization makes the correspondence explicit without relying on the more
detailed Heston application of \cite{Garcia2026Heston}. In particular,
state-dependent covariance appears in the holonomy sector while the resulting
momentum dynamics retain the familiar Riccati closure.

Several extensions are natural, including time-dependent affine coefficients,
jump-affine and more general L\'evy models, and matrix-valued affine state
spaces. Appendix~\ref{app:stock-dependent-default} records a further boundary
case relevant to credit modeling: affine killing remains within the present
framework, whereas an inverse-power equity intensity produces a momentum-shift
coupling and lies beyond finite-dimensional affine closure. A systematic
extension of AHGQ to such nonlocal transform dynamics is left for future work.

\clearpage
\appendix


\section{Additional geometric structure}
\label{app:geometric-action}

The Poincar\'e--Cartan formulation of AHGQ retains the phase-space action
structure familiar from GAQ and provides a connection with path-integral
formulations of derivative pricing.

\subsection{Affine phase-space action and path-integral interpretation}
\label{sec:appendix-affine-action}

The noncentral part of the affine Poincar\'e--Cartan form \(\Theta\),
defined in Section~\ref{sec:poincare-cartan}, is

\begin{equation}
\Theta_A=\frac12\mathbf p^T d\mathbf x-\frac12\mathbf x^T d\mathbf p
+C_A(\mathbf x,\mathbf p)\,dt.
\label{eq:app-pc-form}
\end{equation}

Along a time-parametrized phase-space curve
\((\mathbf x(t),\mathbf p(t))\),

\begin{equation}
\Theta_A=\mathscr L_A\,dt,\qquad
\mathscr L_A
=
\frac12\mathbf p^T\dot{\mathbf x}
-\frac12\mathbf x^T\dot{\mathbf p}
+C_A(\mathbf x,\mathbf p).
\label{eq:app-phase-lagrangian}
\end{equation}

Here a dot denotes differentiation with respect to time. The function
\(\mathscr L_A\) is the phase-space Lagrangian associated with
\(\Theta_A\). Treating \(\mathbf x\) and \(\mathbf p\) as independent
variables, the Euler--Lagrange equations are

\begin{equation}
\frac{d}{dt}\frac{\partial\mathscr L_A}{\partial\dot{\mathbf x}}
-\frac{\partial\mathscr L_A}{\partial\mathbf x}=0,\qquad
\frac{d}{dt}\frac{\partial\mathscr L_A}{\partial\dot{\mathbf p}}
-\frac{\partial\mathscr L_A}{\partial\mathbf p}=0.
\label{eq:app-euler-lagrange}
\end{equation}

Substitution of \eqref{eq:app-phase-lagrangian} gives

\begin{equation}
\dot{\mathbf p}
=
\nabla_{\mathbf x}C_A(\mathbf x,\mathbf p),\qquad
\dot{\mathbf x}
=
-\nabla_{\mathbf p}C_A(\mathbf x,\mathbf p),
\label{eq:app-hamilton-characteristics}
\end{equation}

which are precisely the phase-space components of the characteristic flow
generated by \(X_\Theta\). Along this flow,

\begin{equation}
\mathscr L_A\big|_{X_\Theta} = C_A -\frac12\mathbf p^T\nabla_{\mathbf p}C_A -\frac12\mathbf x^T\nabla_{\mathbf x}C_A = -\eta_A
\label{eq:app-characteristic-lagrangian}
\end{equation}

where $\eta_A$  is defined in \eqref{eq:etaA}. Since
\(\Theta=d\zeta/\zeta+\Theta_A\) and \(\Theta(X_\Theta)=0\), the positive
central coordinate along a characteristic satisfies

\begin{equation}
\frac{d}{dt}\log\zeta
=
-\mathscr L_A\big|_{X_\Theta}
=
\eta_A,\qquad
\zeta(t)
=
\zeta(0)
\exp\!\left(
-\int_0^t
\left.\mathscr L_A(s)\right|_{X_\Theta}\,ds
\right).
\label{eq:app-central-action}
\end{equation}

Thus, along an affine characteristic, the central \(\Rplus\) coordinate
records the exponential of the accumulated Poincar\'e--Cartan phase-space
action.

The phase-space Lagrangian \eqref{eq:app-phase-lagrangian} also connects
AHGQ with path-integral pricing. Since \(C_A\) is quadratic in momentum,
the momentum functional integral is formally Gaussian and may be reduced
to a configuration-space action. For Heston, with the appropriate midpoint
prescription, this yields the standard path-integral Lagrangian and an
exactly solvable radial harmonic-oscillator problem~\cite{Lemmens}.

\clearpage
\section{State-dependent default intensities}
\label{app:stock-dependent-default}

\subsection{Affine intensity driven by an interest-rate state}
\label{sec:affine-rate-intensity}

Let $x$ be a one-dimensional nonnegative rate state and suppose that both the
short rate and the default intensity are affine functions of $x$,

\begin{equation}
r(x)=r_0+r_1x,\qquad \lambda(x)=\lambda_0+\lambda_1x,
\qquad \lambda_0\ge 0,\quad \lambda_1\ge 0.
\label{eq:affine-rate-default-intensity}
\end{equation}

For a zero-recovery claim the pre-default generator contains the total killing
rate $r(x)+\lambda(x)$. Hence

\begin{equation}
r(x)+\lambda(x)=(r_0+\lambda_0)+(r_1+\lambda_1)x.
\label{eq:affine-total-killing}
\end{equation}

This is still affine. In the notation of Section~\ref{sec:matrix-decomposition},
adding the intensity simply changes the killing coefficients according to

\begin{equation}
c\longmapsto c+\lambda_0,\qquad
d\longmapsto d+\lambda_1,
\label{eq:affine-default-killing-shift}
\end{equation}

In the multidimensional case, if
$\lambda(\mathbf x)=\lambda_0+\boldsymbol\lambda^T\mathbf x$, the
corresponding replacement is $c\mapsto c+\lambda_0$ and
$\mathbf d\mapsto\mathbf d+\boldsymbol\lambda$. The state-independent
part $F$ and the Riccati symbol $\mathbf R$ therefore remain of exactly the form
used in the main text. No momentum translation is generated, and the ordinary
finite-dimensional affine transform remains valid.

A particularly transparent example is obtained when the rate state itself is a
CIR short rate,

\begin{equation}
dx_t=\kappa(\theta-x_t)\,dt+\sigma\sqrt{x_t}\,dW_t,
\qquad r_t=x_t,
\label{eq:cir-rate-default-state}
\end{equation}

and the issuer has affine intensity
$\lambda(x)=\lambda_0+\lambda_1x$. The zero-recovery pre-default pricing
operator is

\begin{equation}
\mathcal L_D^{\rm CIR}
=\kappa(\theta-x)\partial_x
+\frac12\sigma^2x\partial_x^2
-\lambda_0-(1+\lambda_1)x.
\label{eq:cir-affine-default-operator}
\end{equation}

The corresponding symbol separates as

\begin{equation}
C_D(x,p)=F_D(p)+xR_D(p),\qquad
F_D(p)=\kappa\theta p-\lambda_0,\qquad
R_D(p)=\frac12\sigma^2p^2-\kappa p-(1+\lambda_1).
\label{eq:cir-affine-default-symbol}
\end{equation}

Consequently the transform equations remain Riccati,

\begin{equation}
\dot\psi=\frac12\sigma^2\psi^2-\kappa\psi-(1+\lambda_1),\qquad
\dot\phi=\kappa\theta\psi-\lambda_0.
\label{eq:cir-affine-default-riccati}
\end{equation}

\subsection{Inverse-power intensity}
\label{sec:inverse-power-intensity}

The equity case is different. If the state is $x=\log S$, a linear function of
$x$ is neither automatically nonnegative on the full log-price domain nor does it
produce the strong inverse stock--credit dependence typically sought in
convertible-bond and structural-equity applications. 

Let $S$ denote the issuer's stock price and let $S_*>0$ be a reference level.
We consider the default intensity

\begin{equation}
\lambda(S)=\ell\left(\frac{S}{S_*}\right)^{-\beta},
\qquad \ell>0,\quad \beta>0.
\label{eq:inverse-power-intensity}
\end{equation}

Equivalently, if

\begin{equation*}
y=\log\left(\frac{S}{S_*}\right),
\end{equation*}

then

\begin{equation}
\lambda(y)=\ell e^{-\beta y}.
\label{eq:log-inverse-power-intensity}
\end{equation}

The form
\eqref{eq:inverse-power-intensity} has the qualitative behavior required in a
simple credit--equity model: default. Closely related
stock-dependent intensities have long been used in convertible-bond models, see, for example,
\cite{AyacheForsythVetzal2003,KovalovLinetsky2008}.

The factor $e^{-\beta y}$ is not affine.
Consequently it cannot be absorbed into the coefficients $c$ and $\mathbf d$
of the affine symbol

\begin{equation*}
C_A(\mathbf x,\mathbf p)=F(\mathbf p)+\mathbf x^T\mathbf R(\mathbf p).
\end{equation*}

A finite truncation of the exponential series provides a corresponding
hierarchy of approximations,
\begin{equation*}
e^{-\beta y}
\approx
\sum_{n=0}^{N}
\frac{(-\beta)^n}{n!} y^n.
\end{equation*}
For \(N=1\), one recovers the local affine approximation and hence a
first-order momentum equation with ordinary Riccati closure.

A systematic extension of AHGQ from affine symbols to polynomial or more
general analytic state dependence would require a corresponding enlargement
of the momentum representation and is left for future work.

\clearpage

\section{Relations among the cocycle structures in AHGQ}
\label{app:cocycle-relations}

Three distinct cocycles occur in the AHGQ construction.
First, the symplectic group $2$-cocycle $\epsilon_2^s$
of Section~\ref{sec:symplectic-lie-group},
equation~\eqref{eq:symplectic-cocycle},
defines the central extension described there.
Second, the thin-path groupoid $1$-cocycle
$\mathcal H_H[\gamma]$ of
Section~\ref{sec:holonomyaction},
equation~\eqref{eq:affine-holonomy},
is multiplicative by
equation~\eqref{eq:holonomy-composition}.
Third, the scalar amplitude derived from
momentum polarization in
Sections~\ref{sec:momentum-polarization}--
\ref{sec:flows_momentum_propagator}
defines a Riccati-flow $1$-cocycle.

The momentum characteristic flows  in \eqref{eq:momentum-characteristic-system}, describe the evolution of the transform parameters in the exponential-affine pricing propagator, while the associated scalar cocycle determines its state-independent amplitude. Their composition laws express the consistency of affine pricing over successive time horizons.

Let $\mathbf p(t)$ denote the solution of the
momentum characteristic equation
\eqref{eq:momentum-characteristic-system},
with initial condition $\mathbf p(0)=\mathbf p_0$.
The corresponding scalar amplitude is
\begin{equation}
A_t(\mathbf p_0)
=
\exp\left(
\int_0^t F(\mathbf p(\tau))\,d\tau
\right).
\end{equation}
Splitting the integral at $\tau$ gives
\begin{equation}
A_{t+\tau}(\mathbf p_0)
=
A_\tau(\mathbf p_0)\,A_t(\mathbf p(\tau)).
\end{equation}

Consider a full affine characteristic path
$\gamma_A:\tau\mapsto
(\mathbf x(\tau),\mathbf p(\tau))$,
$0\leq \tau\leq t$, with initial momentum
$\mathbf p(0)=\mathbf p_0$.
Its momentum component satisfies
\begin{equation}
d\mathbf p/d\tau \equiv \dot{\mathbf p}= R(\mathbf p).
\label{eq:pdot}
\end{equation}
The affine symbol \eqref{eq:affine-symbol}
and its decomposition in
Section~\ref{sec:matrix-decomposition} give

\begin{equation}
C_s+C_H
=
F(\mathbf p)+\mathbf x^T R(\mathbf p).
\end{equation}

Subtracting $\mathbf x^T\dot{\mathbf p}$ from
both sides and using the characteristic equation \eqref{eq:pdot}, we obtain
\begin{equation}
\begin{aligned}
C_s+C_H-\mathbf x^T\dot{\mathbf p}
&=
F(\mathbf p)
+\mathbf x^T\bigl(R(\mathbf p)-\dot{\mathbf p}\bigr)\\
&=F(\mathbf p)
\qquad\text{on }\gamma_A.
\end{aligned}
\end{equation}

Integrating from $0$ to $t$ separates the reduced
logarithmic amplitude into two path integrals:
\begin{equation}
\int_0^t F(\mathbf p(\tau))\,d\tau
=
\int_{\gamma_A}
\left(C_s -\mathbf x^T \dot{\mathbf p} \right) \, d\tau
+
\int_{\gamma_A}C_H\,d\tau.
\end{equation}

The second integral is
$-\log\mathcal H_H[\gamma_A]$
by the holonomy definition (3).
Define the \emph{momentum--polarization
compensation factor}
\begin{equation}
\mathcal B[\gamma_A]
=
\exp\left(
\int_{\gamma_A}
\left(C_s -\mathbf x^T \dot{\mathbf p} \right) \, d\tau
\right).
\end{equation}
Exponentiating gives
\begin{equation}
A_t(\mathbf p_0)
=
\mathcal B[\gamma_A]\,
\mathcal H_H[\gamma_A]^{-1}.
\label{eq:cocycle-compensation}
\end{equation}

Here $\mathcal H_H[\gamma_A]^{-1}$ denotes the
reciprocal of the positive holonomy multiplier,
not the holonomy of an inverse path.  Note that $\mathcal B[\gamma_A]$ and
$\mathcal H_H[\gamma_A]$ are evaluated along the
full characteristic phase-space path and may
depend on the initial coordinate $\mathbf x_0$.
Their product, however, 
depends only on the momentum trajectory
$\mathbf p(\tau)$, with $\mathbf p(0)=\mathbf p_0$.

The relationship between the affine holonomy and the
reduced scalar-amplitude cocycle established in this appendix
holds only along momentum characteristic flows,
not along arbitrary paths of the AHGQ thin-path groupoid, since for an arbitrary thin path
the relation \eqref{eq:pdot} need not hold.

\mathtoolsset{showonlyrefs=true}

\end{document}